\documentclass[letterpaper, 10 pt, conference]{ieeeconf}  % Comment this line out if you need a4paper

\IEEEoverridecommandlockouts                              % This command is only needed if 
\usepackage{ifpdf}
\usepackage{amsmath,amssymb,amsfonts}
\usepackage{url}
\usepackage{graphics} % for pdf, bitmapped graphics files
\usepackage{epsfig} % for postscript graphics files 
\usepackage[lofdepth,lotdepth]{subfig}
\usepackage{amsmath,amssymb,amsfonts}
\usepackage{graphicx}  
\usepackage{tikz} 
\usepackage{epstopdf}
\usepackage{cite}
\usepackage{algorithm, algorithmic, setspace}
\usepackage{url}
\usepackage{hyperref}
\hypersetup{
	colorlinks=true,
	linkcolor=blue,
	filecolor=blue,
	urlcolor=blue,
	citecolor = blue,
}
\allowdisplaybreaks

\newtheorem{thm}{\bf Theorem}
\newtheorem{lem}{\bf Lemma}
\newtheorem{prop}{\bf Proposition}

\newtheorem{assum}{\bf Assumption}
\newtheorem{rem}{\bf Remark}

\DeclareMathOperator{\tr}{tr}  

\DeclareMathOperator{\blkdiag}{blkdiag}

\newcommand{\T}{^\top}  
\newcommand{\SO}{\mathbb{SO}}
\newcommand{\SE}{\mathbb{SE}}

\newcommand{\I}{\mathcal{I}}
\newcommand{\B}{\mathcal{B}}
\newcommand{\C}{\mathcal{C}} 
\newcommand{\PB}{{}^\B} 
\newcommand{\PC}{{}^\C}

\newcommand{\ie}{\textit{i.e.}}  
 
\newcommand*{\doi}[1]{\texttt{DOI:}\href{https://doi.org/\detokenize{#1}.pdf}{\detokenize{#1}}}

\title{\LARGE \bf
Pose, Velocity and Landmark Position Estimation Using IMU and Bearing Measurements
}

\author{Miaomiao Wang and Abdelhamid Tayebi
	\thanks{M. Wang is with the School of Artificial Intelligence and Automation, Huazhong University of Science and Technology, Wuhan 430074, China. Email: {\tt\small mmwang@hust.edu.cn}}
	\thanks{A. Tayebi is with the Department of Electrical Engineering, Lakehead University, Thunder Bay, ON P7B 5E1, Canada, and also with the Department of Electrical and Computer Engineering, Western University, London, ON N6A 3K7, Canada. E-mail: {\tt\small atayebi@lakeheadu.ca}. }
}% 

\begin{document}

\maketitle
\thispagestyle{empty}
\pagestyle{empty}

%%%%%%%%%%%%%%%%%%%%%%%%%%%%%%%%%%%%%%%%%%%%%%%%%%%%%%%%%%%%%%%%%%%%%%%%%%%%%%%%
\begin{abstract}

This paper investigates the estimation problem of the pose (orientation and position) and linear velocity of a rigid body, as well as the landmark positions, using an inertial measurement unit (IMU) and a monocular camera. First, we propose a globally exponentially stable (GES) linear time-varying (LTV) observer for the estimation of body-frame landmark positions and velocity, using IMU and monocular bearing measurements. Thereafter, using the gyro measurements, some landmarks known in the inertial frame and the estimates from the LTV observer, we propose a nonlinear pose observer on $\SO(3)\times \mathbb{R}^3$. The overall estimation system is shown to be almost globally asymptotically stable (AGAS) using the notion of almost global input-to-state stability (ISS). Interestingly, we show that with the knowledge (in the inertial frame) of a small number of landmarks, we can recover (under some conditions) the unknown positions (in the inertial frame) of a large number of landmarks.
Numerical simulation results are presented to illustrate the performance of the proposed estimation scheme.

\end{abstract}

%%%%%%%%%%%%%%%%%%%%%%%%%%%%%%%%%%%%%%%%%%%%%%%%%%%%%%%%%%%%%%%%%%%%%%%%%%%%%%%%
\section{Introduction}

The development of robust estimation algorithms for the pose (attitude and position) and linear velocity of a rigid body is instrumental in robotics and aerospace applications such as unmanned aerial vehicles (UAVs), autonomous ground vehicles and autonomous underwater vehicles. Inertial navigation systems (INSs) can provide estimates of the pose and linear velocity by integrating the kinematic equations of motion using the measurements of an IMU (gyroscope and accelerometer). However, these estimates are not reliable in practical applications involving inaccurate initial conditions, measurement noise and sensor biases. To address this issue, additional sensors, providing partial or full position information, such as Global Positioning Systems (GPSs), range sensors, and vision sensors, are used along with INSs to achieve an accurate and reliable estimation of the pose and linear velocity of the rigid body system in motion. 

Vision-aided INSs, combining an IMU and a vision sensor, have gained significant popularity in recent years, especially in applications in GPS-denied environments. Vision sensors can be categorized into three main categories: monocular cameras, stereo cameras (two cameras in a stereo setup), and RGB-D cameras (RGB camera and a depth sensor). Most of the existing estimation approaches for vision-aided INSs in the literature are of Kalman-type, such as the Extended Kalman filter (EKF) \cite{mourikis2007multi} and the invariant EKF (IEKF) \cite{barrau2017invariant}. Although these Kalman-type filters have been widely implemented in many practical applications, extra care must be taken as these filters offer only local stability guarantees due to the local linearizations. In recent years, nonlinear geometric pose observers, with provable strong stability guarantees, have been proposed in the literature, for instance, \cite{vasconcelos2010nonlinear,hua2011observer,hua2015gradient,wang2018hybrid}. By virtue of the invariance of the pose dynamics on $\SE(3)$ (or $\SO(3)\times \mathbb{R}^3$), gradient-like pose observers on $\SE(3)$ using the group velocity (angular velocity and body-frame linear velocity) and landmark position measurements have been proposed in \cite{hua2011observer,hua2015gradient} with AGAS guarantees, \ie, the estimated pose converges asymptotically to the actual one from almost all initial conditions except from a set of zero Lebesgue measure. Note that $\SE(3)$ (or $\SO(3)$) is a Lie group, which is not homeomorphic to the vector space $\mathbb{R}^n$, and therefore AGAS is the strongest result one can achieve via these time-invariant smooth observers. To achieve global pose estimation, hybrid observers on $\SE(3)$, endowed with global asymptotic stability (GAS) guarantees, have been proposed in \cite{wang2018hybrid}. However, these existing pose observers require the measurements of the linear velocity and the three-dimensional (3D) landmark positions. This requirement presents practical challenges for real-time implementation, particularly in the context of low-cost and small-scale vehicles navigating in GPS-denied environments. 

In practice, obtaining the linear velocity, whether in the inertial frame or the body frame, in GPS-denied environments, is challenging. Alternatively, simultaneous estimation of the pose and linear velocity can be obtained using IMU and landmark position measurements  \cite{barrau2017invariant,hua2018riccati,wang2020nonlinear,wang2020hybrid}. Although the body-frame 3D landmark position measurements can be computed from data collected by a stereo camera or an RGB-D camera, obtaining the landmark position measurements from a monocular camera is not an easy task. Due to the lack of depth information, a monocular camera provides only bearing measurements. To address this challenge, the problem of pose and linear velocity estimation, using IMU and bearing measurements, has been considered, for instance, in \cite{marco2020position,berkane2021nonlinear,wang2022nonlinear}. 
% local Riccati observer in \cite{marco2020position} and AGAS nonlinear observer {wang2022nonlinear}
Furthermore, it is of great importance, from both theoretical and practical perspectives, to develop state observers that can simultaneously estimate the pose and linear velocity as well as the landmark positions using IMU and monocular bearing measurements. In \cite{batista2015navigation}, a Kalman filter has been employed to estimate the position, linear velocity, and landmark ranges, relying on bearing measurements and prior knowledge of the attitude. In \cite{yi2022globally}, the authors addressed the problem of pose, velocity, and landmark position estimation using a parameter estimation-based observer (PEBO).  

In the present work, we consider the estimation problem of the pose and linear velocity of a rigid body as well as the landmark positions, using IMU and monocular bearing measurements. We first propose a GES-LTV observer on the vector space $\mathbb{R}^{3N+6}$ ($N$ denoting the number of landmarks) for the estimation of linear velocity and the landmark positions. Then, based on the estimated velocity and landmark positions as well as some landmarks known in the inertial frame, we propose a nonlinear pose observer on $\SO(3)\times \mathbb{R}^3$. The overall estimation scheme is shown to be AGAS using the notion of almost global ISS \cite{angeli2004almost}.

\section{Mathematical Preliminaries}
\subsection{Notations and Definitions}
The sets of real, non-negative real, and natural numbers are denoted by $\mathbb{R}$, $\mathbb{R}_{\geq 0}$ and $\mathbb{N}$, respectively. We denote by $\mathbb{R}^n$ the $n$-dimensional Euclidean space, and by $\mathbb{S}^{n-1}$ the set of unit vectors in $\mathbb{R}^{n}$. The trace function is denoted by $\tr(\cdot)$, and the Euclidean inner product of two matrices $A, B \in \mathbb{R}^{m\times n}$ is denoted by $\langle \langle A, B \rangle \rangle = \tr(A\T B)$. %, and the Frobenius norm of a matrix $X\in \mathbb{R}^{n\times m}$ is denoted by $\|X\|_F = \sqrt{\tr(X\T X)}$.
Let $I_n$ denote  the $n$-by-$n$ identity matrix and $0_{m\times n}$ denote the $m$-by-$n$ zero matrix.
For a given matrix $A\in \mathbb{R}^{n\times n}$, we define $\mathcal{E}(A)$ as the set of all unit-eigenvectors of $A$ and $\lambda_{m}^A$ as the minimum eigenvalue  of $A$. 
%By $\blkdiag(\cdot)$, we denote the block diagonal matrix. 
The \textit{Special Orthogonal group} of order three, denoted by $\SO(3)$, is defined as
$\SO(3):=\{R\in \mathbb{R}^{3}, RR\T = R\T R=I_3, \det(R)=+1\}$. The \textit{Lie algebra} of $\SO(3)$ is given by $\mathfrak{so}(3):=\{\Omega\in \mathbb{R}^{3}: \Omega = -\Omega\T\}$. 
%The \textit{Special Euclidean group} of order three is defined as $\SE(3):=  \SO(3)\times \mathbb{R}^3$. {\color{red} I don't think that this is correct!!!}
Let $(\cdot)^\times: \mathbb{R}^3 \mapsto \mathfrak{so}(3)$ denote the skew-symmetric map
%\[
%x^\times = \begin{bmatrix}
	%	0 & -x_3 & x_2\\
	%	x_3 & 0 & -x_1\\
	%	-x_2 & x_1 & 0
	%\end{bmatrix},  \forall x=[x_1,x_2,x_3]\T\in \mathbb{R}^3,
%\]
such that $x^\times y = x \times y $, for any $x,y \in \mathbb{R}^3$, with $\times$ denoting the vector cross-product on $\mathbb{R}^3$. For a given matrix $A=[a_{ij}]_{3\times 3} \in \mathbb{R}^{3\times 3}$, we define $\psi: \mathbb{R}^{3\times 3} \to \mathbb{R}^3$ 
\begin{equation} \label{eqn:def_psi}
	\psi(A) := \frac{1}{2}[a_{32}-a_{23}, a_{13}-a_{31}, a_{21}-a_{12}]\T  
\end{equation} 
such that $\langle \langle A, u^\times \rangle\rangle = \tr(A\T u^\times) = 2u\T \psi(A)$ and $\psi(A)=-\psi(A\T)$ for any $A\in \mathbb{R}^{3\times 3}, u\in \mathbb{R}^3$. 
We introduce the following orthogonal projection operator $\pi: \mathbb{S}^2\to \mathbb{R}^{3\times 3}$:
\begin{align}
	\pi(x) = I_3-xx\T, \quad \forall x\in \mathbb{S}^2  \label{eqn:def_pi_x}
\end{align}
which projects any vector in $\mathbb{R}^3$ onto the plane orthogonal to $x$. One can verify that $\pi(x)$ is symmetric positive semi-definite and bounded, $\pi(x) y = 0$ if $x$ and $y$ are collinear, and $\pi^2(x) = \pi(x)$ and $R\pi(x)R\T = \pi(Rx)$ for all $x\in \mathbb{S}^2, R\in \SO(3)$.
For the sake of simplicity, the argument of the time-dependent signals is omitted unless otherwise required for the sake of clarity.    

\section{Problem Formulation}

\subsection{Kinematic Model}
Let $\{\I\}$ and $\{\B\}$ be the inertial and body frames, respectively. The pose of a rigid body is denoted by the pair $(R,p)\in \SO(3)\times \mathbb{R}^3$ where $R$ is the rotation matrix denoting the attitude of the body-attached frame $\{\B\}$ with respect to frame $\{\I\}$, and $p$ is the position of the rigid body expressed in frame $\{\I\}$.
The kinematic equations of a rigid body navigating in 3D space are given by
\begin{subequations}\label{eqn:kinematic}
	\begin{align} 
		\dot{R} & = R \omega^\times  \label{eqn:dot_R} \\
		\dot{p} & = Rv  \label{eqn:dot_p} \\
		\dot{v} & = -\omega^\times v  + R\T g + a \label{eqn:dot_v}
	\end{align}
\end{subequations}
where $\omega,v\in \mathbb{R}^3$ denote the angular velocity and linear velocity of the rigid body expressed in frame $\{\B\}$, respectively, $a\in \mathbb{R}^3$ is the ``apparent acceleration'' capturing all non-gravitational forces applied to the rigid body expressed in frame $\{\B\}$, and $g\in \mathbb{R}^3$ denotes the (constant) gravity vector expressed in frame $\{\I\}$.
We assume that $\omega$ and $a$ are continuously measurable from an onboard IMU and are uniformly bounded.

\subsection{Monocular Bearing Measurements}
We consider a family of $N\geq 3$ landmarks, denoted by $p_i$, in frame $\{\mathcal{I}\}$ for all $i\in \{1,2,\dots, N\}:=\mathbb{I}$.
The landmarks positions are expressed in frame $\{\mathcal{B}\}$ as follows:
\begin{align}
	\PB p_i & = R\T(p_i-p), \quad \forall i\in \mathbb{I} \label{eqn:def_pB}  
\end{align}
Let $\{\mathcal{C}\}$ denote the camera frame attached to the optical center of the camera (see Fig. \ref{fig:pinholemodel}), and $(R_c, p_c)\in \SO(3)\times \mathbb{R}^3$ denote the homogeneous transformation pose from frame $\{\mathcal{B}\}$ to frame $\{\mathcal{C}\}$.  
Then, the $i$-th landmark position measurement expressed in frame $\{\mathcal{C}\}$ is given as follows:
\begin{align} 
	\PC p_i & = R_c^\top (\PB p_i-p_c) . \label{eqn:def_pC}
\end{align} 
The bearing measurement of the $i$-th landmark in frame $\{\mathcal{C}\}$, denoted by $z_i \in \mathbb{S}^2$, is given as follows:
\begin{align}
	z_i :=  \frac{\PC p_i}{\|\PC p_i\|} =  \frac{ R_c^\top (\PB p_i-p_c)}{\|\PB p_i-p_c\|}\label{eqn:bearing_def}
\end{align}
for all $i\in \mathbb{I}$. In practical applications, the pose $(R_c, p_c)$ is constant and can be obtained from an appropriate offline Camera-IMU calibration (see, for instance, \cite{wu2021simultaneous}). 
\begin{figure}[!ht]
	\centering
	\includegraphics[width=0.7\linewidth]{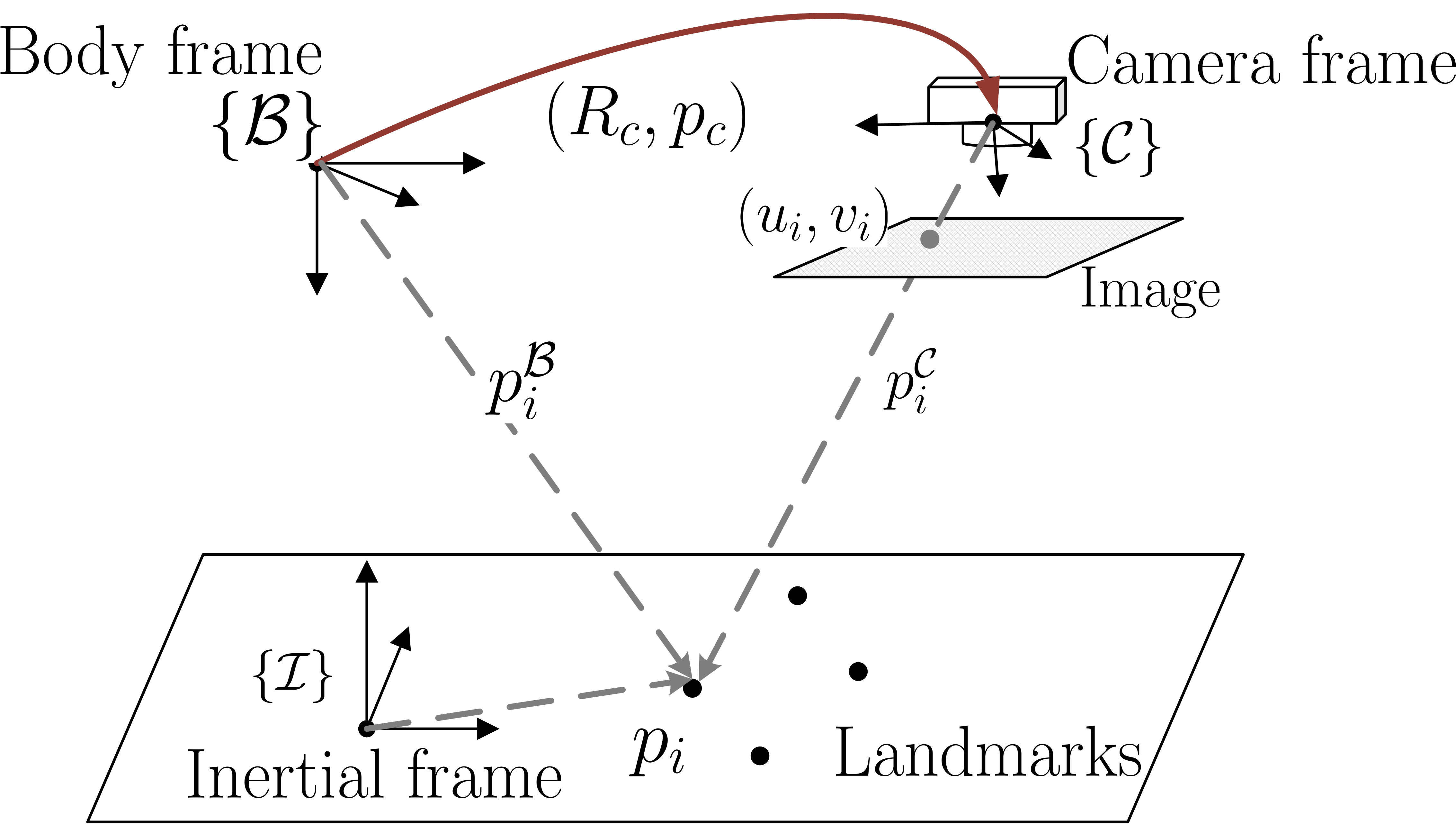}
	\caption{Pinhole model of a monocular camera and coordinate systems of the inertial frame, body frame and camera frame.}
	\label{fig:pinholemodel}
\end{figure}

\begin{rem}
	In practice, the bearing measurements $z_i$ can be constructed directly from the visual measurements of a monocular camera. Let $(u_i,v_i)$ denote the pixel measurements of the $i$-th landmark. Suppose that the camera is well calibrated, the visual measurement can be modeled from the pinhole camera model as follows:    
	\begin{align}
		z'_i
		:=\begin{bmatrix}
			u_i \\
			v_i \\
			1
		\end{bmatrix}
		= \mathcal{K}  \begin{bmatrix}
			\PC p_{iX}/ \PC p_{iZ} \\
			\PC p_{iY}/ \PC p_{iZ}\\
			1
		\end{bmatrix}
		= \frac{1}{\PC p_{iZ}} \mathcal{K}  \PC p_i 
		%	\begin{bmatrix}
			%		p_{iX}^{\mathcal{C}} \\
			%		p_{iY}^{\mathcal{C}} \\
			%		p_{iZ}^{\mathcal{C}}
			%	\end{bmatrix} 
		%	= \frac{1}{p_{iZ}^{\mathcal{C}}} \mathcal{K}  R_c^\top (p_c-p_i^\B)
		\label{eqn:pixel_measurement}
	\end{align}
	where $\PC p_{i} = [\PC p_{iX}, \PC p_{iY}, \PC p_{iZ}]\T$ is the camera-frame landmark position defined in \eqref{eqn:def_pC}, $\PC p_{iZ}$ is known as the depth between the $i$-th landmark and the camera, and $\mathcal{K} $ denotes the camera intrinsic matrix (including the focal length and principal points). Then, the bearing measurement $z_i$ defined in \eqref{eqn:bearing_def} can be computed directly from the visual measurement $z_i'=[u_i, v_i, 1]\T$ as follows:
	\begin{align}
		z_i = %\frac{ p_i^{\mathcal{C}}}{\| p_i^{\mathcal{C}}\|}=  
		\frac{\mathcal{K}^{-1}z'_i}{\|\mathcal{K}^{-1}z'_i\|}, \quad \forall i\in \mathbb{I}.  \label{eqn:def_bearing_pixel}
	\end{align} 
\end{rem}

From \eqref{eqn:bearing_def}, we consider the following  modified output for each  bearing measurement $z_i, i\in \mathbb{I}$ as
\begin{align}
	y_i := \pi(R_c z_i) p_c   = \Pi_{z_i} \PB p_i  \label{eqn:def_y_i}
\end{align}
where $\Pi_{z_i}:= \pi(R_c z_i)$, and we have made use of the facts $\pi(R_c z_i) = R_c \pi(z_i) R_c\T$ and $\pi(R_c z_i)   (\PB p_i-p_c) =0$. %The following assumption ensures the persistency of excitation (PE) condition on each bearing measurement. 
Due to the lack of depth information in the measurements of a monocular camera (\ie, the monocular-bearing measurement $z_i$), the following persistency of excitation (PE) condition on each bearing measurement is needed for the estimation of landmark positions. 
\begin{assum}\label{assum:PE}
	For each bearing measurement $z_i,   i\in \mathbb{I}$, there exist constants $\delta_o, \mu_o>0$ such that
	\begin{equation}\label{eqn:PE}
		\int_{t}^{t+\delta_o} \pi(z'_i(\tau))   d \tau> \mu_o I_3 
	\end{equation}
	for all $t\geq 0$ with $z_i':= R R_cz_i \in \mathbb{S}^2$ and $\pi(\cdot)$ defined in \eqref{eqn:def_pi_x}.
\end{assum}
\begin{rem}
	From the definitions of $\Pi_{z_i}$ in \eqref{eqn:def_y_i} and $\pi(\cdot)$ in \eqref{eqn:def_pi_x}, one obtains $\pi(z'_i) = \pi(R R_cz_i) = R \Pi_{z_i} R\T$. %Therefore, one can rewrite \eqref{eqn:PE} as $\int_{t}^{t+\delta_o} R(\tau) \Pi_{z_i}(\tau) R\T(\tau) d \tau> \mu_o I_3$. 
	From \eqref{eqn:def_pB}-\eqref{eqn:bearing_def}, the modified bearing vector $z_i'$ can be rewritten as $z_i' = R R_c \PC p_i / \|R R_c \PC p_i\|$ with $R R_c \PC p_i = p_i -p -Rp_c$. This implies that $z_i'$ is a unit vector expressed in the inertial frame, which points to the position of the $i$-th landmark from the position of the camera. Therefore, Assumption \ref{assum:PE} indicates that the camera (attached to the rigid body) is in motion and is not indefinitely moving in a straight line passing through any landmark. Moreover, similar to the conditions in \cite[Lemma 3]{wang2022nonlinear} with continuous bearing vector $z_i'(t)$, Assumption \ref{assum:PE} also implies that there exist constants $\epsilon, \delta_o > 0$ such that for any $t\geq 0$, there exists some time $t^* \in (t,t+\delta_o)$ ensuring that $\|z_i'(t) \times z_i'(t^*)\|\geq \epsilon$. 
\end{rem}

The main objective of this work is to design a nonlinear estimation scheme to simultaneously estimate the pose and linear velocity of a rigid body system in \eqref{eqn:kinematic} as well as the landmark positions, with AGAS guarantees, using the IMU measurements ($\omega, a$) and the monocular bearing measurements ($z_1, z_2, \dots, z_N$).

\section{Observer Design}\label{sec:observer_design}

In this section, we propose a nonlinear estimation scheme consisting of an LTV observer on $\mathbb{R}^{3N+6}$ for the body-frame linear velocity and the landmark positions and a nonlinear pose observer on $\SO(3)\times \mathbb{R}^3$. The LTV observer relies on the IMU and monocular bearing measurements, and the pose observer relies on the angular velocity measurements and the estimated body-frame linear velocity and landmark positions. The structure of the proposed estimation scheme is shown in Fig. \ref{fig:estimation_scheme}.
\begin{figure}[!ht]
	\centering
	\includegraphics[width=0.95\linewidth]{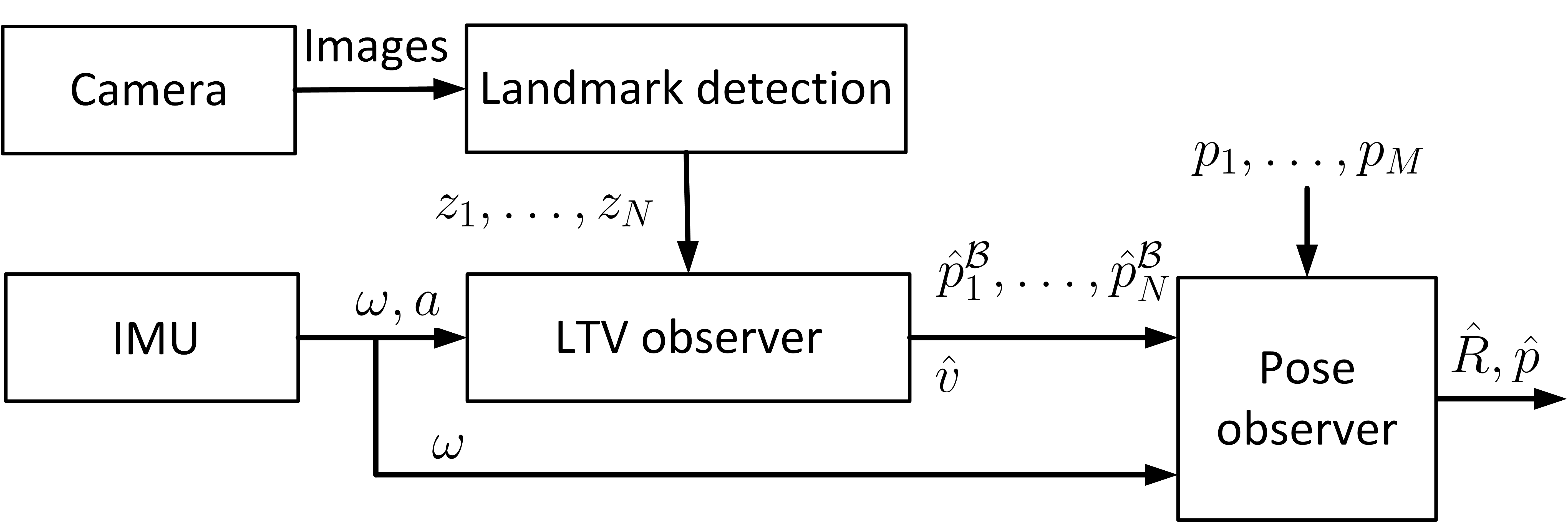}
	\caption{Structure of the proposed estimation scheme for pose, velocity and landmark position estimation.}
	\label{fig:estimation_scheme}
\end{figure}

\subsection{Landmark Position and Velocity Estimation} 
%To simplify the dynamics of the linear velocity in \eqref{eqn:dot_v}, we consider the body-frame gravity vector, denoted by $\eta=R\T g$, as an extended state to estimate. 
Consider the body-frame gravity vector, denoted by $\eta=R\T g$, as an additional state. 
From the dynamics \eqref{eqn:kinematic}, the body-frame landmark position $\PB p_i$ in \eqref{eqn:def_pB} and the modified output in \eqref{eqn:def_y_i}, one obtains the following dynamic system: % of the velocity and landmark positions:
\begin{align}\label{eqn:velocity_landmark}  
	\begin{cases}
		\PB\dot{p}_i  = -\omega^\times \PB p_i -  v \\		
		~~\dot{v} ~  = -\omega^\times v  +  \eta  + a \\
		~~\dot{\eta} ~ = -\omega^ \times \eta \\
		~~y_i  = \Pi_{z_i} \PB p_i 
	\end{cases}
\end{align}
for all $i\in \mathbb{I}$. Define  $x_L:= [\PB p_1\T, \PB p_2\T, \dots, \PB p_N\T]\T \in \mathbb{R}^{3N}$ and $y: =[y_1\T, y_2\T, \dots, y_N\T]\T \in \mathbb{R}^{3N}$ as the vectors of body-frame landmark positions and modified output measurements, respectively.
From \eqref{eqn:velocity_landmark}, the dynamics of the landmark positions and the output measurements can be written as 
\begin{align}\label{eqn:p_v_eta}
	\begin{cases}
		\dot{x}_L = -A_\omega x_L- \Gamma  v  	\\ 
		~y ~ = \Lambda_z x_L
	\end{cases}
\end{align}
where the matrices $A_\omega\in \mathbb{R}^{3N \times 3N}$, $\Lambda_z  \in \mathbb{R}^{3N\times 3N}$ and   $\Lambda \in \mathbb{R}^{3N \times 3}$  are defined as follows:
\begin{subequations} \label{eqn:def_A_Gamma_Lambda}
	\begin{align}
		A_\omega &:= \blkdiag(\omega^\times, \dots, \omega^\times), ~~
		\Gamma  := [I_3, \dots, I_3]\T,  \\
		\Lambda_z &:= \blkdiag(\Pi_{z_1}, \Pi_{z_2},, \dots,\Pi_{z_N}). %\in \mathbb{R}^{3N\times 3N}.
	\end{align}  
\end{subequations} 
Letting $x := [x_L\T, v\T, \eta\T]\T \in \mathbb{R}^{3N+6}$ as the extended state, one can rewrite \eqref{eqn:velocity_landmark} and \eqref{eqn:p_v_eta} as the following LTV system:
\begin{align} \label{eqn:ltv}
	\begin{cases}
		\dot{x} = A x +  B a  \\
		y = Cx
	\end{cases}
\end{align}
with the matrices $A\in \mathbb{R}^{(3N+6)\times (3N+6)}, B\in \mathbb{R}^{(3N+6)\times 3}$ and $C\in \mathbb{R}^{3N\times (3N+6)}$ defined as
\begin{subequations}\label{eqn:ltv_ABC}
	\begin{align}
		A &= \begin{bmatrix}
			-A_\omega & -\Gamma  & 0_{3N\times 3} \\
			0_{3 \times 3N} & -\omega^\times & I_3 \\
			0_{3 \times 3N} & 0_{3\times 3} & -\omega^\times
		\end{bmatrix},  B = \begin{bmatrix}
			0_{3N\times 3} \\
			I_{3}  \\
			0_{3 \times 3}  
		\end{bmatrix} \\
		C & = \begin{bmatrix}
			\Lambda_z & 0_{3N \times 3} & 0_{3N \times 3}
		\end{bmatrix}.
	\end{align}
\end{subequations}    
Note that the pair $(A, C)$ is time-varying, where the matrix $A$ relies on the angular velocity $\omega$, and the matrix $C$ related to the time-varying bearing measurements ($z_1,z_2,\dots, z_N$).
\begin{rem} 
	The introduction of the body-frame gravity vector $\eta = R\T g$ as an additional state, similar to the work in \cite{wang2020nonlinear}, which plays a crucial role in two aspects: 1) it decouples the dynamics of the landmark positions and linear velocity from the attitude, which leads to a nice simplified LTV system \eqref{eqn:ltv}; 2) it enables handling the situations where the gravity vector $g$ is unknown or not precisely known.  
\end{rem}

This subsection aims to design an observer to estimate the body-frame landmark positions and the body-frame linear velocity of the rigid body for the LTV system \eqref{eqn:ltv}. Let $\hat{x}_L:= [\PB \hat{p}_1\T, \dots, \PB \hat{p}_N\T]\T$ denote the estimate of the vector of the body-frame landmark positions $x_L$, $\hat{v}$ denote the estimate of the body-frame linear velocity $v$, and $\hat{\eta}$ denote the estimate of the body-frame gravity vector $\eta$. For the sake of simplicity, define $\hat{x} := [\hat{x}_L\T, \hat{v}\T, \hat{\eta}\T]\T $ as the estimate of the state $x$. Then, we propose the following linear Riccati observer:
\begin{align}
	\dot{\hat{x}} &= A \hat{x} +  B a + K(y-C\hat{x}) \label{eqn:observer_ltv}  
\end{align}
with $\hat{x}(0)\in \mathbb{R}^{3N+6}$. The gain matrix $K$ is designed as:
\begin{align} \label{eqn:def_K}
	K = PC\T Q
\end{align} 
where $P$ is the solution to the continuous Riccati equation (CRE)
\begin{align} 
	\dot{P} &= AP + PA\T - PC\T QCP + V \label{eqn:CRE}
\end{align}
with $P(0) \in \mathbb{R}^{(3N+6)\times (3N+6)}$ being symmetric positive definite, $Q \in \mathbb{R}^{3N \times 3N} $ and $V \in \mathbb{R}^{(3N+6)\times (3N+6)}$ being continuous, bounded and uniformly positive definite.

\begin{lem}\label{lem:UO}
	Suppose that $\omega(t)$ is uniformly continuous and bounded, and Assumption \ref{assum:PE} holds. Then, the pair $(A, C)$ defined in \eqref{eqn:ltv_ABC} is uniformly observable. 
\end{lem}
\begin{proof}
	See Appendix \ref{sec:UO}.
\end{proof}
\begin{rem}
	Note that the matrix $A(t)$ is continuous and uniformly bounded since $\omega$ is continuous and uniformly bounded by assumption. Moreover, the matrix $C(t)$ is also continuous and uniformly bounded since $z_i$ is continuous and $\Pi_{z_i}$ is naturally bounded by definition. Therefore, according to \cite{wang2020nonlinear},   
	Lemma \ref{lem:UO} ensures the existence of the solution $P(t)$ to the CRE \eqref{eqn:CRE} that satisfies $p_m I_{3N+6}< P(t) < p_M I_{3N+6}$ for all $t\geq 0$ with some constant  scalars $0<p_m<p_M<\infty$.
\end{rem}

Define $\tilde{x}:= x -\hat{x} $ as the estimation error, and let $\tilde{x}=[\tilde{x}_L\T, \tilde{v}\T, \tilde{\eta}\T]\T $  with $\tilde{v}:= v - \hat{v}$, $\tilde{\eta}:= \eta - \hat{\eta}$ and $\PB \tilde{p}_i:= \PB p_i  
 - \PB \hat{p}_i$ denoting the estimation errors of velocity, gravity vector and landmark positions, respectively. Then, one obtains the following closed-loop system
\begin{align}
	\dot{\tilde{x}} &= (A-KC) \tilde{x}   \label{eqn:error_ltv}  
\end{align}
with $(A, C)$ defined in \eqref{eqn:ltv_ABC} and $K$ defined in \eqref{eqn:def_K}.  
\begin{prop}\label{prop:x_GES}
	Consider the closed-loop system \eqref{eqn:error_ltv} with the pair $(A, C)$ defined in \eqref{eqn:ltv_ABC} and the gain matrix $K$ defined in \eqref{eqn:def_K}. Suppose that the conditions in Lemma \ref{lem:UO} are satisfied, and $Q(t), V(t)$ are chosen to be continuous, bounded, and uniformly positive definite. Then, the equilibrium $(\tilde{x}=0)$ is uniformly globally exponentially stable, \ie, there exist constant scalars $\alpha, \lambda>0$ such that $\|\tilde{x}(t)\| \leq \alpha \exp(-\lambda t) \|\tilde{x}(0)\|$ for all $t\geq 0$.
\end{prop}
\begin{proof} 
	See Appendix \ref{sec:prop_x_GES}.
\end{proof}

% \begin{rem}
From \eqref{eqn:velocity_landmark} and \eqref{eqn:p_v_eta}, the proposed LTV observer \eqref{eqn:observer_ltv} can be explicitly rewritten as follows:
\begin{subequations}
	\label{eqn:observer_PV1}
	\begin{align} 
		\PB\dot{\hat{p}}_i & = -\omega^\times \PB\hat{p}_i -  \hat{v} +  K_i (y-C\hat{x})\\
		\dot{\hat{v}}~     & = -\omega^\times \hat{v}  + \hat{\eta} +  a  + K_v (y-C\hat{x})\\
		\dot{\hat{\eta}}~  & = -\omega^\times \hat{\eta} +  K_\eta (y-C\hat{x})
	\end{align}
\end{subequations}
for all  $i\in \mathbb{I}$ with $K = [K_1\T, \cdots, K_N\T, K_v\T, K_\eta\T] \T $ and $K_i, K_v, K_\eta \in \mathbb{R}^{3\times 3N}$.	 
% \end{rem} 

\subsection{Pose Estimation}
In this subsection, we consider the problem of pose estimation relying on the gyro measurements and the estimated velocity and body-frame landmark positions, assuming that $M$ ($3 \leq M\leq N$) landmarks are known in the inertial frame. Without loss of generality, suppose that the first $M$ landmarks, \ie, $\mathbb{I}':=\{1,2,\dots, M\}$, are known and satisfy the following assumption:
\begin{assum}\label{assum:pose}
	Among the $M \geq 3$ landmarks known in the inertial frame, there exist at least 3 landmarks that are non-aligned.
\end{assum} 
\begin{rem}
	Assumption \ref{assum:pose} requires a minimum of 3 non-aligned landmarks known in the inertial frame, which is a common assumption in pose estimation problems using 3D landmark position measurements, see for instance \cite{hua2015gradient,wang2018hybrid}.
\end{rem}

Let $\hat{R}$ and $\hat{p}$ denote the estimates of the attitude $R$ and position $p$, respectively. We propose the following pose observer on $\SO(3)\times \mathbb{R}^3$:
\begin{subequations}\label{eqn:observer_pose}
	\begin{align}
		\dot{\hat{R}} & = \hat{R}(\omega  + k_R \hat{R}\T \sigma_R)^\times \\
		\dot{\hat{p}} & = \hat{R} \hat{v} + (k_R \sigma_R)^\times (\hat{p}-p_o) + k_p \sigma_p
	\end{align}
\end{subequations}
where $(\hat{R}(0), \hat{p}(0))\in \SO(3)\times \mathbb{R}^3$ and $p_o := \sum_{i\in \mathbb{I}'} \rho_i p_i$ denotes the weighted center of the $M$ landmarks with the scalars  $0<\rho_i<1$ and $\sum_{i\in \mathbb{I}'}\rho_i = 1$. The innovation terms $\sigma_R$ and $\sigma_p$ are defined as:
\begin{align}\label{eqn:def_innovation} 
	\sigma_R    = \frac{1}{2} \sum_{i\in \mathbb{I}'}\rho_i \nu_i \times \xi_i, \quad %(p_i-p_c)^\times (p_i-\hat{p}-\hat{R}~\PB\hat{p}_i)  \\
	\sigma_p    =   \sum_{i\in \mathbb{I}'}\rho_i  \xi_i %(p_i-\hat{p}-\hat{R}~\PB \hat{p}_i) 
\end{align}
with $\nu_i := p_i - p_o$ and $\xi_i  := p_i-\hat{p}-\hat{R}~\PB \hat{p}_i$.
Let $\tilde{R} = R\hat{R}\T$ and $\tilde{p} =   p - \tilde{R} \hat{p} - (I_3 - \tilde{R})p_o$ denote the estimation errors of the attitude and the position, respectively. One verifies that $(\hat{R}=R, \hat{p} = p)$ when  $(\tilde{R}=I_3, \tilde{p} = 0)$. %Let $\PB \tilde{p}_i := \PB \hat{p}_i - \PB p_i  $ denote the landmark position estimation error. 
Then, from \eqref{eqn:def_pB} and the definition of $\tilde{p}$, one has
$ \hat{R}~ \PB \hat{p}_i = \tilde{R}\T(p_i - p) + \hat{R}~ \PB \tilde{p}_i$ and $\xi_i = \tilde{R}\T \tilde{p} + (I_3 - \tilde{R}\T) \nu_i + \hat{R} ~\PB \tilde{p}_i$ with $\PB \tilde{p}_i = \PB \hat{p}_i - \PB p_i  $ denoting the landmark position estimation error. Then, one obtains  
\begin{subequations}\label{eqn:innovation_error}
\begin{align} 
	\sigma_R & =  \psi(\mathsf{M} \tilde{R})+  \frac{1}{2} \sum_{i\in \mathbb{I}'}\rho_i \nu_i^\times \hat{R} ~\PB \tilde{p}_i \\
	\sigma_p   
	%& =   \sum_{i\in \mathbb{I}'}\rho_i \left( \tilde{R}\T \tilde{p} + (I_3 - \tilde{R}\T) (p_i - p_c) \right) +  \sum_{i\in \mathbb{I}'}\rho_i \hat{R}~ \PB \tilde{p}_i \nonumber \\ 
	& =   \tilde{R}\T \tilde{p}  + \hat{R} \sum_{i\in \mathbb{I}'}\rho_i  \PB \tilde{p}_i 
\end{align}
\end{subequations}
where $\mathsf{M} := \sum_{i\in \mathbb{I}'}\rho_i \nu_i \nu_i\T$, $\psi(\mathsf{M}  \tilde{R}) = -\frac{1}{2} \sum_{i\in \mathbb{I}'}\rho_i \nu_i ^\times  \tilde{R}\T \nu_i $ with $\psi(\cdot)$ defined in \eqref{eqn:def_psi}, and we used the facts $\sum_{i\in \mathbb{I}'}\rho_i  = 1$, $p_o=\sum_{i\in \mathbb{I}'}\rho_i p_i$ and $\nu_i ^\times \nu_i =0$. Let $\bar{\mathsf{M} }: = \frac{1}{2}(\tr(\mathsf{M} )I_3 - \mathsf{M} ) = -\sum_{i\in \mathbb{I}'} \rho_i (\nu_i^\times )^2$. By Assumption \ref{assum:pose}, it is always possible to tune the scalars $\rho_i \in (0,1), \forall i\in \mathbb{I}'$ such that $\bar{\mathsf{M} }$ is positive definite with three distinct eigenvalues \cite{tayebi2013inertial}. 
Then, from \eqref{eqn:kinematic}, \eqref{eqn:observer_pose} and \eqref{eqn:innovation_error} as well as the definition of $\tilde{x}$, one obtains the following pose error dynamics:
\begin{subequations}\label{eqn:error_pose}  
	\begin{align} 
		\dot{\tilde{R}}  
		& =  \tilde{R}  ( -k_R \psi(\mathsf{M} \tilde{R}) + \Gamma_1\tilde{x} )^\times  \\
		\dot{\tilde{p}}  
		& = - k_p \tilde{p}  +   \Gamma_2\tilde{x}  
	\end{align}
\end{subequations}
where the matrices $\Gamma_1, \Gamma_2 \in \mathbb{R}^{3\times (3N+6)}$ are defined as
\begin{align} 
	\Gamma_1  &: = \frac{k_R}{2} [\rho_1\nu_1^\times \hat{R},\cdots,  \rho_M\nu_M^\times \hat{R},   0_{3\times 3(N-M+2)}], \nonumber \\
	\Gamma_2 &: =  [k_p\rho_1{R},\cdots, k_p\rho_M{R}, 0_{3\times 3(N-M)},R,0_{3\times 3}]. \label{eqn:def_Gamma}
\end{align}
Note that the matrices $\Gamma_1$ and $\Gamma_2$ are time-varying and satisfy $\|\Gamma_1\|_F =  \frac{k_R}{2} ({\sum_{i=1}^{M}2\rho_i^2\|\nu_i\|^2})^{\frac{1}{2}}  := k_{\Gamma_1}$ and $\|\Gamma_2\|_F =    ({3(k_p^2\sum_{i=1}^{M}\rho_i^2+1)})^{\frac{1}{2}} :=k_{\Gamma_2}$ for all $t\geq 0$, \ie,  $\forall \hat{R},R \in  \SO(3)$. Consequently, it follows that $\|\Gamma_s \tilde{x}\|\leq k_{\Gamma_s}\|\tilde{x}\|$ for all $t\geq 0$ and $s\in\{1,2\}$. The introduction of the weighted center $p_o$ in the observer design is motivated by \cite{wang2018hybrid,wang2020nonlinear,wang2020hybrid} to decouple the position error dynamics and the attitude error dynamics in \eqref{eqn:error_pose}.

\begin{prop}\label{prop:ISS}
	Consider system \eqref{eqn:error_pose} on $\SO(3)\times \mathbb{R}^3 \times  \mathcal{D}$ with $\Gamma_1,\Gamma_2$ defined in \eqref{eqn:def_Gamma}. Suppose that $\mathcal{D}$ is a closed subset of $\mathbb{R}^{3N+6}$, and choose $k_R, k_p >0$ and $\bar{\mathsf{M} }$ being positive definite with three distinct eigenvalues. Then, system \eqref{eqn:error_pose} is almost globally ISS with respect to the equilibrium $(I_3, 0)$.
\end{prop}

\begin{proof}
	See Appendix \ref{sec:ISS}. 
\end{proof}

% \begin{rem}
% 	s
% \end{rem}
Now, one can state the following main result:
\begin{thm}\label{thm:AGAS}
	Consider the closed-loop system \eqref{eqn:error_ltv} and \eqref{eqn:error_pose} with $\Gamma_1,\Gamma_2$ defined in \eqref{eqn:def_Gamma}, the pair $(A, C)$ defined in \eqref{eqn:ltv_ABC} and gain matrix $K$ designed in \eqref{eqn:def_K}. Suppose that the IMU measurements are continuous and uniformly bounded, and Assumptions \ref{assum:PE}-\ref{assum:pose} hold. Choose $k_R, k_p >0$, $\bar{\mathsf{M} }$ being positive definite with three distinct eigenvalues, and  $Q(t), V(t)$ being continuous, bounded, and uniformly positive definite. 
	Then, the equilibrium $(\tilde{R}=I_3,\tilde{p}=0,\tilde{x}=0)$ of the closed-loop system is AGAS. 
\end{thm}
\begin{proof}
	The proof of Theorem  \ref{thm:AGAS} is motivated by \cite[Theorem 2]{angeli2004almost} using the results obtained in Proposition \ref{prop:x_GES} and \ref{prop:ISS}. 
	From \eqref{eqn:error_ltv} and \eqref{eqn:error_pose}, one obtains the following cascaded system on $\SO(3)\times \mathbb{R}^3 \times \mathbb{R}^{3N+6}$
	\begin{align}\label{eqn:cascaded_system}
		\begin{cases}
		\dot{\tilde{R}} =  \tilde{R}  (  -k_R \psi(\mathsf{M} \tilde{R}) +   \Gamma_1\tilde{x} )^\times  \\
		~\dot{\tilde{p}}  = -    k_p   \tilde{p}  +   \Gamma_2\tilde{x}  \\		
		~\dot{\tilde{x}}  = (A-KC) \tilde{x}
		\end{cases}
	\end{align}
	It is clear that $(\tilde{R}=I_3,\tilde{p}=0,\tilde{x}=0)$ is the equilibrium of the $(\tilde{R},\tilde{p})$-subsystem and $(\tilde{x}=0)$ is the equilibrium of the $\tilde{x}$-subsystem. The local exponential stability of the equilibrium $(\tilde{R}=I_3,\tilde{p}=0,\tilde{x}=0)$ for system \eqref{eqn:cascaded_system} can be easily deduced following the proof of \cite[Theorem 1]{wang2022nonlinear}, therefore, it is omitted here. As per Proposition \ref{prop:x_GES}, the  $\tilde{x}$-subsystem is GES at $(\tilde{x}=0)$, \ie, $\|\tilde{x}(t)\| \leq \alpha \exp(-\lambda t) \|\tilde{x}(0)\|$ for all $t\geq 0$ with some constant scalars $\alpha,\lambda>0$. Consequently, it follows that $\|\Gamma_s(t) \tilde{x}(t)\|\leq k_{\Gamma_s}\alpha \exp(-\lambda t) \|\tilde{x}(0)\|$ for all $t\geq 0$ and $s\in\{1,2\}$.  
    Moreover, as per Proposition \ref{prop:ISS}, the $(\tilde{R},\tilde{p})$-subsystem is almost globally ISS with respect to the equilibrium $(\tilde{R}=I_3,\tilde{p}=0)$ and input $\tilde{x}$. Therefore, according to the proof of \cite[Theorem 2]{angeli2004almost}, one can conclude that the cascaded system \eqref{eqn:cascaded_system}  is AGAS at the equilibrium $(\tilde{R}=I_3,\tilde{p}=0, \tilde{x}=0)$. %, \ie, the estimated states $(\hat{R},\hat{p},\hat{x})$ converge asymptotically to the actual states $(R,p,x)$ from almost all initial conditions except from a set of zero Lebesgue measure.  
    This completes the proof. 
\end{proof}

\begin{rem}
	Theorem \ref{thm:AGAS} implies that the estimated states $(\hat{R},\hat{p},\hat{v},\hat{\eta},\PB \hat{p}_{i})$ converge asymptotically to the actual states $(R,p,v,\eta,\PB p_i)$ from almost all initial conditions except from a set of zero Lebesgue measure.  %the estimated pose $(\hat{R},\hat{p})$, (body-frame) linear velocity $\hat{v}$, and (body-frame) landmark positions $\PB \hat{p}_{i}$ converge to the actual ones from almost all initial conditions except from a set of zero Lebesgue measure.
    Since $\SO(3)$ is a Lie group and not homeomorphic to $\mathbb{R}^n$, AGAS is the strongest result one can achieve via the time-invariant smooth observer \eqref{eqn:observer_pose}. 
	Note that the inertial frame velocity and gravity vector can be obtained as $\hat{R}\T \hat{v}$ and $\hat{R}\hat{\eta}$, respectively. Moreover, the inertial frame positions of the unknown landmarks ($M+1,M+2,\hdots, N$) can be obtained from the estimated pose $(\hat{R},\hat{p})$ and landmark position measurements $\PB \hat{p}_i$. From \eqref{eqn:def_pB}, one can obtain the estimated inertial frame landmark positions as $\hat{p}_i = \hat{R}\PB \hat{p}_i + \hat{p}$ for all $i=M+1,\hdots,N$. This feature can be potentially extended to mapping problems to map the environment with at least three non-aligned landmarks known in the inertial frame. 
\end{rem}

\section{SIMULATION}

In this section, numerical simulation results are presented to illustrate the performance of the proposed observers in Section \ref{sec:observer_design}. 
We consider an autonomous UAV equipped with an IMU and a downward-facing monocular camera that captures $N=16$ landmarks located on the ground. Within this setup, we assume that the first $M=4$ non-aligned landmarks are known in the inertial frame.
We further consider that the UAV is moving along a `8'-shape trajectory described by $p(t) = 2[\sin(t), \sin(t)\cos(t), 1]\T$ (m) with initial conditions $p(0)=[0,0,2]\T,v(0)=[2,2,0]\T, R(0)=I_3$, and the IMU measurements are chosen as $\omega(t) = [-\cos(2t), 1, \sin(2t)]$ (rad/s) and $a(t) = R\T(t) (\ddot{p}(t) - g)$ (m/s\textsuperscript{2}) with the gravity vector $g=[0,0,-9.81]\T$. Moreover, the transformation pose between the camera and the rigid body is given as $(R_c = I_3, p_c = [0.02, 0.06, 0.01]\T)$, and the bearing measurements are computed using \eqref{eqn:bearing_def} in terms of the ground-truth pose and landmark positions.

The initial conditions for the proposed observers are chosen as $\hat{v}(0)=\hat{\eta}(0)=\PB \hat{p}_i(0)=[0,0,0]\T, \forall i\in \mathbb{I}$ for the LTV observer \eqref{eqn:observer_ltv}, and $\hat{p}(0)=[0,0,0]\T, \hat{R}(0)= \exp(0.9\pi u^\times), u \in \mathbb{S}^2$ for the pose observer \eqref{eqn:observer_pose}. The gain parameters are given as $k_R = 40, k_p = 100, \rho_i = 1/M, \forall i\in \mathbb{I}'$ for the pose observer \eqref{eqn:observer_pose} and $Q = 10^{-4}I_{3N}, V = 10^{6} I_{3N+6}$ for the LTV observer \eqref{eqn:observer_ltv}. Simulation results are shown in Fig. \ref{fig:simulation}. Figure \ref{fig:fig1} illustrates the time evolution of landmark positions and velocity estimation errors as well as $\hat{R}\hat{\eta}$ which converges to $g$.  Figure \ref{fig:fig2} illustrates the convergence to zero of the pose estimation errors from a large initial condition. Furthermore, Figure \ref{fig:fig3} provides a visual representation of the locations of the landmarks and the trajectories of the UAV's position, the estimated position and the estimated (inertial frame) landmark positions.% $\hat{p}_i = \hat{R}\PB \hat{p}_i + \hat{p}$ for all  $i\in \mathbb{I}$.

\begin{figure}[!ht]
	\centering
	\subfloat[]{\includegraphics[width=0.4\textwidth]{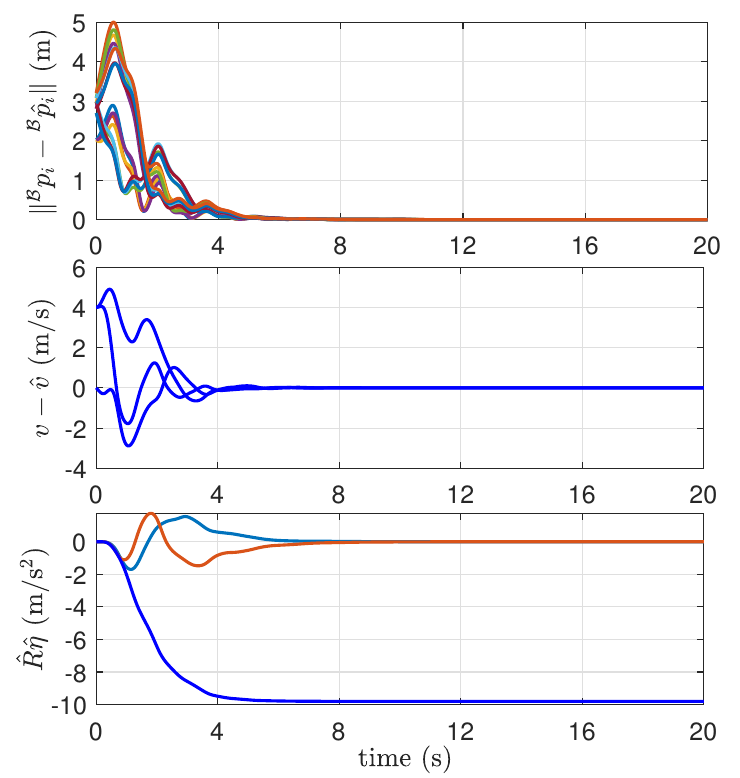}
	\label{fig:fig1}}\\[-0.01in]
	\subfloat[]{\includegraphics[width=0.4\textwidth]{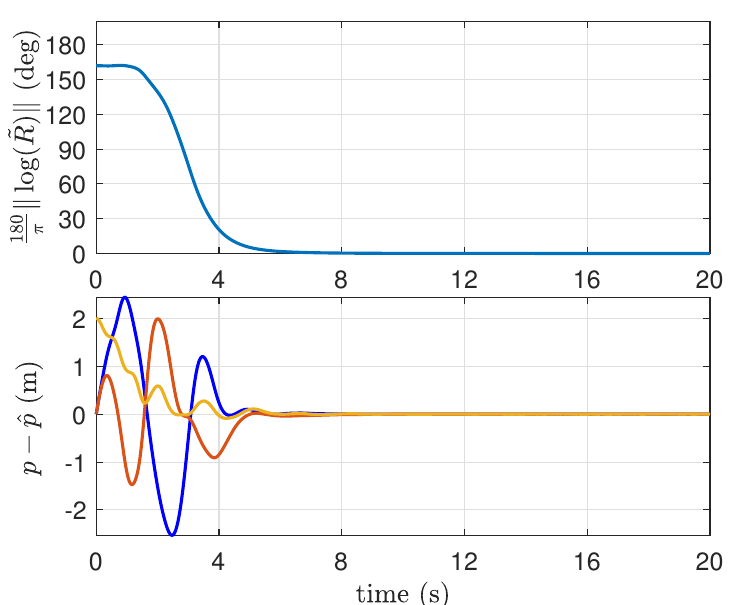}
	\label{fig:fig2}}\\[-0.01in]
    \subfloat[]{\includegraphics[width=0.4\textwidth]{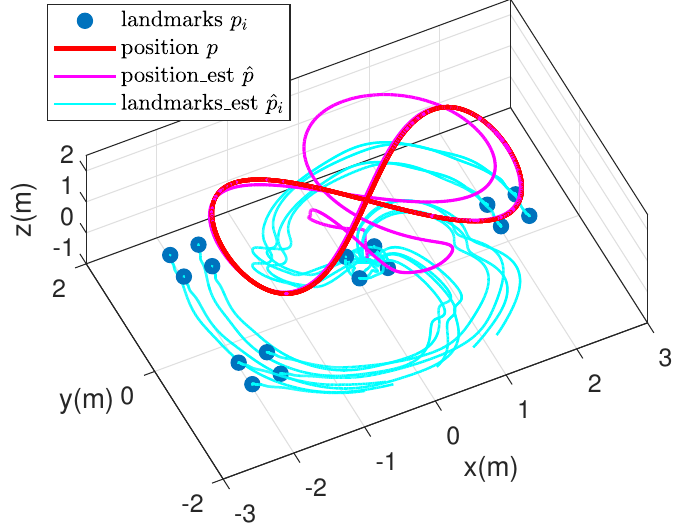}
	\label{fig:fig3}}\\[-0.03in] 
	\caption{Estimation results of the proposed estimation scheme: (a) Estimation errors of the linear velocity and landmark positions as well as the estimated gravity vector; (b) Estimation errors of the attitude and position; (c) The locations of the landmarks and the trajectories of the position, the estimated position and the estimated (inertial frame) landmark positions $\hat{p}_i = \hat{R}\PB \hat{p}_i + \hat{p}, \forall i\in \mathbb{I}$.} 
	\label{fig:simulation}
\end{figure}

% \begin{figure}[!ht]
% 	\centering
% 	\includegraphics[width=0.82\linewidth]{fig3} 
% 	\caption{The locations of the landmarks $p_i$ and the trajectories of the position $p$, the estimated position $\hat{p}$ and the estimated (inertial frame) landmark positions $\hat{p}_i = \hat{R}\PB \hat{p}_i + \hat{p}$.}
% 	\label{fig:pose}
% \end{figure}

\section{CONCLUSIONS}
In this work, we formulated the problem of pose, velocity, and landmark position estimation in terms of the IMU and monocular bearing measurements. We first designed an LTV observer for the body-frame linear velocity and landmark position estimation on the vector space $\mathbb{R}^{3N+6}$. With these estimated body-frame linear velocity and landmark positions as well as some landmarks known in the inertial frame, we designed a nonlinear pose observer on $\SO(3)\times \mathbb{R}^3$. The overall closed-loop system on $\SO(3)\times \mathbb{R}^3 \times \mathbb{R}^{3N+6}$, which is an interconnection of a linear GES sub-system and an almost global ISS sub-system on $\SO(3)\times \mathbb{R}^3$, is shown to be AGAS. Our estimation results showed that the depth for each landmark can be recovered when the PE condition \eqref{eqn:PE} on each bearing measurement is satisfied. This is a nice feature for practical low-cost applications when the vehicles are only equipped with an IMU and a monocular camera. The simulation results also showed that the proposed estimation scheme can be potentially extended to map the environment with a small number of landmarks (at least three non-aligned) known in the inertial frame.

                                  % the text height of the last page by a suitable amount.
                                  % sure that you do not shorten the text height too much.

%%%%%%%%%%%%%%%%%%%%%%%%%%%%%%%%%%%%%%%%%%%%%%%%%%%%%%%%%%%%%%%%%%%%%%%%%%%%%%%%

%%%%%%%%%%%%%%%%%%%%%%%%%%%%%%%%%%%%%%%%%%%%%%%%%%%%%%%%%%%%%%%%%%%%%%%%%%%%%%%%

%%%%%%%%%%%%%%%%%%%%%%%%%%%%%%%%%%%%%%%%%%%%%%%%%%%%%%%%%%%%%%%%%%%%%%%%%%%%%%%%
\section*{APPENDIX}


\subsection{Proof of Lemma \ref{lem:UO}} \label{sec:UO}
This proof is motivated by the proof of \cite[Lemma 3]{wang2020hybrid} and \cite[Lemma 2]{wang2022nonlinear}. 
In order to show that the pair $(A(t), C(t))$ in \eqref{eqn:ltv_ABC} is uniformly observable, we are going to verify the existence  of constants $\delta,\mu>0$ such that the observability Gramian satisfies
\begin{align}
	W_o(t,t+\delta)  
	&  =\frac{1}{\delta} \int_{t}^{t+\delta} \Phi\T(\tau,t) C\T(\tau)   C(\tau) \Phi(\tau,t) d\tau   \nonumber \\
	&\geq \mu I_{3N+6}, \quad \forall t\geq 0  \label{eqn:W_oA_C}
\end{align} 
with $\Phi(t,\tau)$ denoting the state transition matrix associated to $A(t)$ satisfying: $\frac{d}{dt} \Phi(t,\tau)= A(t) \Phi(t,\tau)$,  $\Phi(t,t) = I$,  $\Phi^{-1}(t,\tau) = \Phi(\tau,t)$ for all $t,\tau\geq 0$, and   $ \Phi(t_3,t_2)\Phi(t_2,t_1) = \Phi(t_3,t_1)$ for all $t_1,t_2,t_3\geq 0$.

 %The key step is to get the explicit state transition matrix associated to $A(t)$ such that we can simplify the analysis of the observability Gramian matrix $W_o(t_j,t+\varGamma)$. 
To find the explicit solution of the observability Gramian matrix, one can rewrite the matrices $A(t)$ and $C(t)$ in \eqref{eqn:ltv_ABC} as follows:
\begin{align}
	A & = \underbrace{\begin{bmatrix}
		\omega^\times & \dots & 0_{3\times 3} \\
		\vdots & \ddots & \vdots \\
		0_{3\times 3} & \dots & \omega^\times
	\end{bmatrix}}_{\Omega} + \underbrace{
		\begin{bmatrix}
			0_{3N \times 3N} & -\Gamma  & 0_{3N\times 3} \\
			0_{3 \times 3N} & 0_{3\times 3} & I_3 \\
			0_{3 \times 3N} & 0_{3\times 3} & 0_{3\times 3}
		\end{bmatrix}
	}_{\bar{A}}  \nonumber \\
	C &= \Lambda_z \underbrace{\begin{bmatrix}
		I_{3N} & 0_{3N\times 3} & 0_{3N\times 3}
	\end{bmatrix}}_{\bar{C}} \label{eqn:A_decomp}  %\label{eqn:C_decomp} 
\end{align}  
with $\Gamma$ and $\Lambda_z$ defined in \eqref{eqn:def_A_Gamma_Lambda}. 
For each $0<n\in \mathbb{N}$,  we introduce the following block diagonal matrix:
\begin{equation*}
	\mathcal{T}_n(\bar{R}): = \blkdiag(\bar{R}, \bar{R},\dots, \bar{R})   \in \mathbb{R}^{3n\times 3n}, ~ \forall \bar{R}\in \SO(3).
\end{equation*}
It is not difficult to verify that $T\T T = T T\T = I_{3n}$ if $T = \mathcal{T}_n(\bar{R})$.  
Let $\mathsf{T} := \mathcal{T}_{N+2} (R\T) \in \mathbb{R}^{(3N+6)\times (3N+6)}$ and $\bar{\mathsf{T}}:= \mathcal{T}_{N} (R\T) \in \mathbb{R}^{3N\times 3N}$ where $R$ is the rotation matrix of the rigid body. 
Then, one obtains the following properties:
\begin{align} \label{eqn:properties_T}
		\mathsf{T}(t) \bar{A}          = \bar{A} \mathsf{T}(t), \quad   
		\bar{C} {\mathsf{T}}(t)                  =\bar{\mathsf{T}}(t) \bar{C}.   %\label{eqn:barT_1}\\
\end{align}
% for all $t\geq 0$.
%From Lemma \ref{lemm:Phi_A}, 
Similar to the proof of \cite[Lemma 3]{wang2020hybrid},  from \eqref{eqn:A_decomp}-\eqref{eqn:properties_T} and the fact $\dot{\mathsf{T}}(t) = -\Omega(t) \mathsf{T}(t)$,  the state transition matrix $\Phi(t,\tau)$ can be explicitly written as
\begin{equation}
	\Phi(t,\tau) = \mathsf{T}(t) \bar{\Phi}(t,\tau) \mathsf{T}\T(\tau)
	\label{eqn:def_Phi}
\end{equation}
with $\bar{\Phi}(t,\tau) = \exp(\bar{A}(t-\tau))$ denoting the state transition matrix associated to  $\bar{A}$. Consequently, one can show that
\begin{align}
	C(\tau) \Phi(\tau,t) %& = C(\tau) \mathsf{T}(\tau) \bar{\Phi}(\tau,t) \mathsf{T}\T(t) \nonumber\\
	                   & = \Lambda_z(\tau) \bar{C}\mathsf{T}(\tau) \bar{\Phi}(\tau,t) \mathsf{T}\T(t) \nonumber \\
	                   & = \Lambda_z(\tau)\bar{\mathsf{T}}(\tau) \bar{C} \bar{\Phi}(\tau,t) \mathsf{T}\T(t).
\end{align}
 
For the sake of simplicity, let $\check{C}(t) :=  \bar{\mathsf{T}}\T(t)  \Lambda_z(t) \bar{\mathsf{T}}(t) \bar{C}  $ such that %one has 
$\bar{C}\T \bar{\mathsf{T}}\T(\tau) \Lambda_z\T (\tau) \Lambda_z(\tau) \bar{\mathsf{T}}(\tau) \bar{C} 
= \check{C}\T(\tau)  \check{C}(\tau) $
from the fact
$\bar{\mathsf{T}}   \bar{\mathsf{T}}\T  = I_{3N}$. Then, the observability Gramian matrix in \eqref{eqn:W_oA_C} associated to the pair $(A(t),C(t))$, can be written as
\begin{align}
	 & W_o(t,t+\delta) \nonumber\\
	%  & ~   =  \frac{1}{\delta} \int_{t}^{t+\delta}\mathsf{T}(t) \bar{\Phi}\T(\tau,t) \bar{C}\T \bar{\mathsf{T}}\T(\tau) \Lambda_z\T (\tau) \Lambda_z(\tau) \bar{\mathsf{T}}(\tau) \bar{C} 	\bar{\Phi}(\tau,t)  \mathsf{T}\T(t)  \nonumber\\
	 & ~   =  \mathsf{T}(t)  \left( \frac{1}{\delta}
		\int_{t}^{t+\delta}  \bar{\Phi}\T(\tau,t) \check{C}\T(\tau)  \check{C}(\tau)
	\bar{\Phi}(\tau,t) d\tau \right)  \mathsf{T}\T(t)  \nonumber  %\label{eqn:W_oTWT}
\end{align}
Let $
\bar{W}_o(t,t+\bar{\delta}) := \frac{1}{\bar{\delta}} \int_{t}^{t+\bar{\delta}}
\bar{\Phi}\T(\tau,t)
\check{C}\T(\tau)  \check{C}(\tau)
\bar{\Phi}(\tau,t) d \tau$ denote the observability Gramian matrix associated to the pair $(\bar{A},\check{C}(t))$. 
Suppose that the pair $(\bar{A}, \check{C}(t))$ is uniformly observable, \ie, there exist constants $\bar{\delta}, \bar{\mu}>0$ such that $\bar{W}_o(t,t+\bar{\delta}) \geq \bar{\mu} I_{3N+6}$ for all $t\geq 0$. Then, picking $\delta\geq \bar{\delta}$ and $0<\mu\leq \frac{\bar{\mu} \bar{\delta}}{\delta}$, one can show that
$ 
	W_o(t,t+\delta)   
	%&\geq  \frac{1}{\delta} \mathsf{T}(t)  \big(\int_{t}^{t+\delta'} \bar{\Phi}\T(\tau,t)
	% \check{C}\T(\tau)  \check{C}(\tau) 	\bar{\Phi}(\tau,t) d \tau  \big)  \mathsf{T}\T(t)  \\
	\geq  \frac{\bar{\delta}}{\delta} \mathsf{T}(t)  \bar{W}_o(t,t+\bar{\delta}) \mathsf{T}\T(t)  
	 \geq \frac{ \bar{\mu} \bar{\delta}}{\delta} \mathsf{T}(t)  \mathsf{T}\T(t)\geq \mu I_{3N+6}
$ for all $t\geq 0$, 
which implies that the pair $(A(t),C(t))$   in \eqref{eqn:ltv_ABC} is uniformly observable. It remains to prove uniform observability of the pair $(\bar{A}, \check{C}(t))$.

%Next, we are going to show that the pair $(\bar{A},\check{C}(t))$ is UCO.
From the definition of $\bar{A}$ in \eqref{eqn:A_decomp}, one can verify that  $\bar{A}$ is a
nilpotent matrix with $\bar{A}^3 =0$. Applying \cite[Lemma 1]{wang2022nonlinear}, one only needs to show that there exist constants $\bar{\delta}, \bar{\mu}>0$ such that
\begin{equation} \label{eqn:condition_O}
	\int_{t}^{t+\bar{\delta}}  \mathcal{O}\T(\tau) \mathcal{O}(\tau) d\tau > \bar{\mu} I_{3N+6}, \quad \forall t\geq 0
\end{equation}
where the matrix $\mathcal{O} := [\check{C}\T(t), ~ (\check{C}(t) \bar{A})\T, ~  (\check{C}(t) \bar{A}^2)\T]\T \in \mathbb{R}^{9N\times (3N+6)}$ is given as
\begin{align}\label{eqn:def_O}
	\mathcal{O}(t)
	% & = \begin{bmatrix}
	% 		\check{C}(t) \\
	% 		\check{C}(t) \bar{A} \\  
	% 		\check{C}(t) \bar{A}^2 
	% 	\end{bmatrix}   
	= \underbrace{
		\begin{bmatrix}
			\bar{\Lambda}_z(t) & 0_{3N \times 3N}                & 0_{3N \times 3N} \\
			0_{3N \times 3N}                 & \bar{\Lambda}_z(t) & 0_{3N \times 3N}  \\
			0_{3N \times 3N}                & 0_{3N \times 3N}                & \bar{\Lambda}_z(t)
		\end{bmatrix}
	}_{\mathcal{M}(t)} \underbrace{
		\begin{bmatrix}
			\bar{C} \\
			\bar{C} \bar{A} \\
			\bar{C} \bar{A}^2
		\end{bmatrix} 
	}_{{\mathcal{O}}_c} %:= \mathcal{M}(t) {\mathcal{O}}_c
\end{align}
with $\bar{\Lambda}_z: = \bar{\mathsf{T}}\T  \Lambda_z\bar{\mathsf{T}} = \blkdiag(\bar{\Pi}_{z_1}, \bar{\Pi}_{z_2}, \dots, \bar{\Pi}_{z_N})$ from \eqref{eqn:def_A_Gamma_Lambda} and $\bar{\Pi}_{z_i}: = R \Pi_{z_i}R\T$.  
From the definitions of $\Pi_{z_i}$, $z_i'$ and the properties of  $\pi(\cdot)$ defined in \eqref{eqn:def_pi_x}, one can show that $\bar{\Pi}_{z_i}\T \bar{\Pi}_{z_i} = R\Pi_{z_i}\T  \Pi_{z_i}R\T =R  \Pi_{z_i}R\T = \pi(R R_cz_i ) =  \pi(z'_i) $. Therefore, from the definition of $\bar{\Lambda}_z$, inequality \eqref{eqn:PE} in Assumption \ref{assum:PE} implies that 
$
	\int_{t}^{t+\delta_o} \bar{\Lambda}_z\T(\tau) \bar{\Lambda}_z  (\tau) d\tau > \mu_o I_{3N}
$ holds for all $t\geq 0$. 
Consequently, from the definition of $\mathcal{M}$ in \eqref{eqn:def_O}, one obtains
\begin{align}\label{eqn:M^TM}
	\int_{t}^{t+\delta_o} \mathcal{M}\T(\tau) \mathcal{M}(\tau) d\tau > \mu_o I_{9N}, \quad \forall t\geq 0.
\end{align}
Then, picking $\bar{\delta}\geq \delta_o$, it follows from \eqref{eqn:def_O} and \eqref{eqn:M^TM} that
\begin{align}
	\int_{t}^{t+\bar{\delta}}  \mathcal{O}\T(\tau) \mathcal{O}(\tau) d \tau
	 & =	 \int_{t}^{t+\bar{\delta}} {\mathcal{O}}_c\T \mathcal{M}\T(\tau) \mathcal{M}(\tau) {\mathcal{O}}_c d\tau \nonumber\\
	 & \geq   {\mathcal{O}}_c\T \left(  \int_{t}^{t+\delta_o}  \mathcal{M}\T(\tau) \mathcal{M}(\tau) d \tau \right)  {\mathcal{O}}_c  \nonumber \\
	 & \geq \mu_o  {\mathcal{O}}_c\T  {\mathcal{O}}_c. \label{eqn:O_c^TO_c}
\end{align}
Let $\bar{\mathcal{O}}_c := [\bar{C}\T~ N_1\T~ N_2\T]\T$ be a matrix composed of $3N+6$ row vectors of matrix ${\mathcal{O}}_c$ with $N_1$ denoting the first three rows of $\bar{C}\bar{A}$ and $N_2$ denoting the first three rows of $\bar{C}\bar{A}^2$. From the definitions of $\bar{A}$ in \eqref{eqn:A_decomp}  and $\bar{C}$ in \eqref{eqn:C_decomp}, one has
\begin{align*}
	\bar{\mathcal{O}}_c = \begin{bmatrix}
		                    \bar{C} \\
		                    N_1 \\
		                    N_2
	                    \end{bmatrix} = \begin{bmatrix}
		                                    I_{3N}        & 0_{3N \times 3} & 0_{3N \times 3} \\ 
		                                    0_{3\times 3N}  & -I_3          & 0_{3\times 3} \\
		                                    0_{3\times 3N}  & 0_{3\times 3} & -I_3
	                                    \end{bmatrix}
\end{align*}
which indicates that ${\mathcal{O}}_c\T  {\mathcal{O}}_c  \geq  \bar{\mathcal{O}}_c\T \bar{\mathcal{O}}_c = I_{3N+6}$. % matrix $\bar{\mathcal{O}}'$ has full rank of $3N+6$.
Consequently, picking $\bar{\mu} < \mu_o$ it follows from \eqref{eqn:O_c^TO_c} that
\begin{align}
      \int_{t}^{t+\bar{\delta}}  \mathcal{O}\T(\tau) \mathcal{O}(\tau) d \tau \geq \mu_o   {\mathcal{O}}_c\T  {\mathcal{O}}_c > \bar{\mu} I_{3N+6} \nonumber  
	% \bar{z}\T \left(  \int_{t}^{t+\bar{\delta}}  \mathcal{O}\T(\tau) \mathcal{O}(\tau) d \tau \right) \bar{z}
	%  & \geq \mu_o  \bar{z}\T {\mathcal{O}}_c\T  {\mathcal{O}}_c \bar{z} \nonumber \\
	%  & \geq \mu_o  \bar{z}\T \bar{\mathcal{O}}_c\T \bar{\mathcal{O}}_c \bar{z} \nonumber\\
	%  & = \mu_o > \bar{\mu} \nonumber  
\end{align}
for all $t\geq 0$, which gives \eqref{eqn:condition_O}. This completes the proof.

\subsection{Proof of Proposition \ref{prop:x_GES}} \label{sec:prop_x_GES}
Since the conditions in Lemma \ref{lem:UO} are satisfied, and $Q(t), V(t)$ are chosen to be continuous, bounded, and uniformly positive definite, according to \cite{wang2020nonlinear}, there exists a unique solution $P(t)$ to the CRE \eqref{eqn:CRE} satisfying $p_m I_{3N+6}< P(t) < p_M I_{3N+6}$ for all $t\geq 0$ with some constant scalars $0<p_m<p_M<\infty$. 

Consider the following Lyapunov function candidate:
\begin{equation}
	V(\tilde{x}) = \tilde{x}\T P^{-1}\tilde{x}.  \label{eqn:def_V_x}
\end{equation}
From the fact $p_m I_{3N+6}< P < p_M I_{3N+6}$, one has
\begin{equation}
	\frac{1}{p_M} \|\tilde{x}\|^2  \leq V(\tilde{x})  \leq \frac{1}{p_m} \|\tilde{x}\|^2   \label{eqn:V_x_bound}
\end{equation}
for all $t \geq 0$. Moreover, one can show from \eqref{eqn:CRE} that $\dot{P}^{-1} = - P^{-1} \dot{P} P^{-1} = - P^{-1}   A  - A\T P^{-1} + C\T QC - P^{-1}V P^{-1} $. 	
%$\dot{P}^{-1} = - P^{-1} \dot{P} P^{-1} = - P^{-1} ( AP + PA\T - PC\T QCP + V) P^{-1} = $
Since $V(t)$ is uniformly positive definite, there exists a constant $v_m > 0$ such that $V(t) \geq v_m I_{3N+6}$ for all $t\geq 0$. 
Hence, from \eqref{eqn:error_ltv} and \eqref{eqn:def_K}, the time-derivative of $V$ defined in \eqref{eqn:def_V_x} is given as 
\begin{align}
	\dot{V}(\tilde{x}) & = \tilde{x}\T(A-KC)P^{-1}\tilde{x} + \tilde{x}\T P^{-1}(A-KC) \tilde{x} \nonumber \\
	& ~~  + \tilde{x}\T (- P^{-1}A  - A\T P^{-1} + C\T QC - P^{-1}V P^{-1})\tilde{x} \nonumber \\
	& = -\tilde{x}\T C\T Q C \tilde{x} - \tilde{x}\T   P^{-1}V P^{-1} \tilde{x}  \nonumber \\
	& \leq - \frac{v_m}{p_M^2} \|\tilde{x}\|^2, \quad \forall t\geq 0 \label{eqn:dot_V_x}
\end{align} 
where we made use of the fact $-\tilde{x}\T C\T Q C \tilde{x} \leq 0$ since $Q(t)$ is uniformly positive definite. Therefore, from \eqref{eqn:V_x_bound} and \eqref{eqn:dot_V_x}, one can show that $\dot{V}(\tilde{x}) \leq - \lambda V(\tilde{x})$ with $\lambda : = \frac{v_m p_m}{2p_M^2}$, which implies $V(\tilde{x}(t)) \leq \exp(-2\lambda t) V(\tilde{x}(0))$. Consequently, from \eqref{eqn:V_x_bound} one concludes $\|\tilde{x}(t)\| \leq \alpha \exp(-\lambda t) \|\tilde{x}(0)\|$ with $\alpha := ({\frac{p_M}{p_m}})^{\frac{1}{2}}$ for all $t \geq 0$, \ie, the equilibrium point $\tilde{x}=0$ is GES. This completes the proof. 
	 
\subsection{Proof of Proposition \ref{prop:ISS}} \label{sec:ISS}
	The proof is motivated by the results in \cite[Proposition 1]{wang2021nonlinear} and \cite[Proposition 1]{wang2022hybrid} using the approach in \cite{angeli2010stability}. We first show that the error system \eqref{eqn:error_pose} satisfies Assumptions A0–A2 in \cite{angeli2010stability}.  One can easily verify that assumption A0 is fulfilled since system \eqref{eqn:error_pose} evolves on the manifold $\SO(3)\times \mathbb{R}^3\times \mathcal{D}$ with $\mathcal{D}$ being a closed subset of $\mathbb{R}^{3N+6}$. For assumption A1, we consider the smooth function on $\SO(3)\times \mathbb{R}^3$ as
	\begin{equation}
		V(\tilde{R},\tilde{p}) = \tr(\mathsf{M} (I_3 - \tilde{R})) + \frac{1}{2} \|\tilde{p}\|^2 
	\end{equation}
	whose time derivative along the trajectories of \eqref{eqn:error_pose} with $\tilde{x} \equiv 0$ is given as
	\begin{align}
		\dot{V}(\tilde{R},\tilde{p}) &= \tr(-\mathsf{M} \tilde{R}(-k_R \psi(\mathsf{M} \tilde{R}))^\times) + \tilde{p}\T (- k_p   \tilde{p} ) \nonumber \\
		& = -2k_R \|\psi(\mathsf{M} \tilde{R}\|^2 - k_p \|\tilde{p}\|^2
	\end{align}
	where we made use of the facts $\psi(A)=-\psi(A\T)$ and $\tr(A\T u^\times)  = 2u\T \psi(A)$ for any $A\in \mathbb{R}^{3\times 3}, u\in \mathbb{R}^3$. This implies that $\dot{V}(\tilde{R},\tilde{p}) < 0$ for all $(\tilde{R},\tilde{p}) \notin \mathcal{W}:=\{(\tilde{R},\tilde{p})\in \SO(3)\times \mathbb{R}^3, \psi(\mathsf{M} \tilde{R})=0, \tilde{p}=0\}$ with $\mathcal{W}$ denoting the equilibria set of the zero-input system \eqref{eqn:error_pose} (\ie, $\tilde{x}\equiv 0$), and then assumption A1 is fulfilled. Applying LaSalle’s invariance principle, it follows that the solution $(\tilde{R},\tilde{p})$ of the zero-input system \eqref{eqn:error_pose} converges asymptotically to the equilibrium set $\mathcal{W}$. From $\psi(\mathsf{M} \tilde{R}) = 0$ one has $\mathsf{M} \tilde{R} = \tilde{R}\T \mathsf{M} $, which implies that $\tilde{R}\in \{I_3\} \cup \{\tilde{R} \in \SO(3): \tilde{R} =\exp(\pi v^\times), v\in \mathcal{E}(\mathsf{M} )\}$ \cite{mahony2008nonlinear}. Hence, the equilibrium set $\mathcal{W}$ can be explicitly rewritten as 
	$ 
		\mathcal{W} = \{(I_3,0)\}\cup\{(\tilde{R},\tilde{p})\in \SO(3)\times \mathbb{R}^3:  
		\tilde{R} =\exp(\pi v^\times), v\in \mathcal{E}(\mathsf{M} ), \tilde{p}=0\}.
	$ 
	Note that the undesired equilibria of the zero-input system \eqref{eqn:error_pose} in the set $\mathcal{W} \setminus \{(I_3,0)\}$ are isolated since $\bar{\mathsf{M} } = \frac{1}{2}(\tr(\mathsf{M} )I_3 - \mathsf{M} )$ is positive definite with three distinct eigenvalues. One can further show that the linearized zero-input system \eqref{eqn:error_pose} at each undesired equilibrium $(R_{v\in \mathcal{E}(\mathsf{M} )}^* = \exp(\pi v^\times),\tilde{p}=0) \in \mathcal{W} \setminus \{(I_3,0)\}$ has at least one positive eigenvalue (for instance, see the proof of \cite[Theorem 1]{wang2022nonlinear} for the $\tilde{R}$-subsystem). Hence, the equilibrium $(I_3,0)$ is AGAS for system \eqref{eqn:error_pose} with $\tilde{x}\equiv 0$ and the assumption A2 is fulfilled. 

	To show that system \eqref{eqn:error_pose} fulfills the ultimate boundedness property, consider 
	\begin{align}
		W(\tilde{R},\tilde{p}) = |\tilde{R}|_I^2 + \|\tilde{p}\|^2
	\end{align}
	with $|\tilde{R}|_I:= \frac{1}{4}\tr(I_3-\tilde{R}) \in [0,1]$. Then, from \eqref{eqn:error_pose} one has
	\begin{align}
		\dot{W}(\tilde{R},\tilde{p}) & = \frac{1}{4}\tr(-\tilde{R} (-k_R \psi(\mathsf{M} \tilde{R}) + \Gamma_1\tilde{x} )^\times)  \nonumber \\
		&\quad  + 2\tilde{p}\T ( - k_p \tilde{p}+ \Gamma_2\tilde{x}) \nonumber \\
		% & = \frac{1}{2}(-k_R \psi(\mathsf{M} \tilde{R}) +   \Gamma_1\tilde{x} )\T \psi(\tilde{R})  \nonumber \\
		% & \quad - 2 k_p \| \tilde{p}\|^2 + 2\tilde{p}\T \Gamma_2\tilde{x} \nonumber \\   
		& = - \frac{k_R}{2}  \psi(\mathsf{M} \tilde{R}) \T \psi(\tilde{R})  +  \frac{1}{2} (\Gamma_1\tilde{x}) \T \psi(\tilde{R})  \nonumber \\
		& \quad - 2 k_p \| \tilde{p}\|^2 + 2\tilde{p}\T \Gamma_2\tilde{x} \nonumber \\ 
		& \leq  - \frac{k_R\lambda_m^{\bar{\mathsf{M} }}  }{2}  \|\psi(\tilde{R}) \|^2 +  \frac{k_{\Gamma_1} }{2}  \|\tilde{x} \|\|\psi(\tilde{R}) \| \nonumber \\
		& \quad - 2 k_p \| \tilde{p}\|^2 + 2k_{\Gamma_2}\|\tilde{p}\|  \|\tilde{x}\| \nonumber \\ 
		& \leq  - k_R \lambda_m^{\bar{\mathsf{M} }}  |\tilde{R}|_I^2(1-|\tilde{R}|_I^2)  - k_p \| \tilde{p}\|^2  + \delta_u (\|\tilde{x} \|)   \nonumber \\
		%& \quad +  \frac{k_{\Gamma_1} }{2} \|\tilde{x} \|  +   \frac{k_{\Gamma_2}^2}{k_p}   \|\tilde{x}\|^2 \nonumber \\ 
		& \leq  - k_R \lambda_m^{\bar{\mathsf{M} }}  |\tilde{R}|_I^2   - k_p \| \tilde{p}\|^2  + 2k_R \lambda_m^{\bar{\mathsf{M} }}  + \delta_u (\|\tilde{x} \|)  \nonumber \\ 
		& \leq -\alpha W(\tilde{R},\tilde{p}) + 2k_R \lambda_m^{\bar{\mathsf{M} }}  + \delta_u (\|\tilde{x} \|) 
	\end{align}
	where $\alpha: = \min\{k_R \lambda_m^{\bar{\mathsf{M} }}, k_p\} > 0, \delta_u (\|\tilde{x} \|) :=  (\frac{k_{\Gamma_1}^2}{4k_R\lambda_m^{\bar{\mathsf{M} }}}  +   \frac{k_{\Gamma_2}^2}{k_p} )  \|\tilde{x}\|^2$, and we have made use of the properties of $\psi(\cdot)$ and the facts $\psi(\mathsf{M} \tilde{R}) \T \psi(\tilde{R})  = \psi(\tilde{R})\T \bar{\mathsf{M} } \psi(\tilde{R}) \geq \lambda_m^{\bar{\mathsf{M} }} \|\psi(\tilde{R})\|^2 $, $\|\psi(\tilde{R})\|^2 = 4|\tilde{R}|_I^2(1-|\tilde{R}|_I^2)$, $\|\Gamma_1\tilde{x}\| \leq k_{\Gamma_1} \|\tilde{x}\|$ and $\|\Gamma_2 \tilde{x}\|\leq k_{\Gamma_2}\|\tilde{x}\|$, and the inequalities $ \frac{k_{\Gamma_1} }{2}  \|\tilde{x} \|\|\psi(\tilde{R}) \| \leq \frac{k_R\lambda_m^{\bar{\mathsf{M} }}}{4}   \|\psi(\tilde{R}) \|^2 + \frac{k_{\Gamma_1}^2}{4k_R\lambda_m^{\bar{\mathsf{M} }}} \|\tilde{x}\|^2$ and $2k_{\Gamma_2}\|\tilde{p}\|  \|\tilde{x}\| \leq k_p\tilde{p}\|^2 +  \frac{k_{\Gamma_2}^2}{k_p}  \|\tilde{x}\|^2 $. Applying the results in \cite[Proposition 2 and Proposition 3]{angeli2010stability}, one can conclude that system \eqref{eqn:error_pose} satisfies the ultimate boundedness property, and therefore it is almost globally ISS with respect to the equilibrium $(I_3,0)$. This completes the proof.

%\section*{ACKNOWLEDGMENT}

%%%%%%%%%%%%%%%%%%%%%%%%%%%%%%%%%%%%%%%%%%%%%%%%%%%%%%%%%%%%%%%%%%%%%%%%%%%%%%%%

\bibliographystyle{IEEEtran}
\bibliography{mybib}

\end{document}